\documentclass[12pt]{article}

\def\UseBibLatex{1}

\makeatletter
\def\input@path{{styles/}}
\makeatother

\providecommand{\BibLatexMode}[1]{}
\providecommand{\BibTexMode}[1]{}

\renewcommand{\BibLatexMode}[1]{#1}
\renewcommand{\BibTexMode}[1]{}

\ifx\UseBibLatex\undefined%
  \renewcommand{\BibLatexMode}[1]{}
  \renewcommand{\BibTexMode}[1]{#1}
\fi

\BibLatexMode{%
   \usepackage[bibencoding=utf8,style=alphabetic,backend=biber]{biblatex}%
   \usepackage{sariel_biblatex}%
}

\usepackage[cm]{fullpage}%
\usepackage{amsmath}%
\usepackage{amssymb}%
\usepackage[table]{xcolor}%

\usepackage[amsmath,thmmarks]{ntheorem}%

\usepackage{titlesec}%
\usepackage{xcolor}%
\usepackage{mleftright}%
\usepackage{xspace}%
\usepackage{graphicx}
\usepackage{hyperref}%
\usepackage[inline]{enumitem}
\usepackage{hyperref}%
\usepackage[ocgcolorlinks]{ocgx2}
\usepackage{qed_duck}

\hypersetup{%
      unicode,
      breaklinks,%
      colorlinks=true,%
      urlcolor=[rgb]{0.25,0.0,0.0},%
      linkcolor=[rgb]{0.5,0.0,0.0},%
      citecolor=[rgb]{0,0.2,0.445},%
      filecolor=[rgb]{0,0,0.4},
      anchorcolor=[rgb]={0.0,0.1,0.2}%
}

\titlelabel{\thetitle. }%

\theoremseparator{.}%

\theoremstyle{plain}%
\newtheorem{theorem}{Theorem}[section]

\newtheorem{lemma}[theorem]{Lemma}

\newtheorem{corollary}[theorem]{Corollary}

\theoremstyle{plain}%
\theoremheaderfont{\sf} \theorembodyfont{\upshape}%
\newtheorem*{remark:unnumbered}[theorem]{Remark}%
\newtheorem{remark}[theorem]{Remark}%

\newtheorem{defn}[theorem]{Definition}

\theoremheaderfont{\em}%
\theorembodyfont{\upshape}%
\theoremstyle{nonumberplain}%
\theoremseparator{}%
\theoremsymbol{\myqedsymbol}%
\newtheorem{proof}{Proof:}%

\providecommand{\emphind}[1]{}%
\renewcommand{\emphind}[1]{\emph{#1}\index{#1}}

\definecolor{blue25emph}{rgb}{0, 0, 11}

\providecommand{\emphic}[2]{}
\renewcommand{\emphic}[2]{\textcolor{blue25emph}{%
      \textbf{\emph{#1}}}\index{#2}}

\providecommand{\emphi}[1]{}%
\renewcommand{\emphi}[1]{\emphic{#1}{#1}}

\definecolor{almostblack}{rgb}{0, 0, 0.3}

\providecommand{\emphw}[1]{}%
\renewcommand{\emphw}[1]{{\textcolor{almostblack}{\emph{#1}}}}%

\providecommand{\emphOnly}[1]{}%
\renewcommand{\emphOnly}[1]{\emph{\textcolor{blue25emph}{\textbf{#1}}}}

\providecommand{\myqedsymbol}{\rule{2mm}{2mm}}

\newcommand{\HLink}[2]{\hyperref[#2]{#1~\ref*{#2}}}
\newcommand{\HLinkSuffix}[3]{\hyperref[#2]{#1\ref*{#2}{#3}}}

\newcommand{\figlab}[1]{\label{fig:#1}}
\newcommand{\figref}[1]{\HLink{Figure}{fig:#1}}

\newcommand{\thmlab}[1]{{\label{theo:#1}}}
\newcommand{\thmref}[1]{\HLink{Theorem}{theo:#1}}

\newcommand{\corlab}[1]{\label{cor:#1}}
\newcommand{\corref}[1]{\HLink{Corollary}{cor:#1}}%

\newcommand{\lemlab}[1]{\label{lemma:#1}}
\newcommand{\lemref}[1]{\HLink{Lemma}{lemma:#1}}%

\newcommand{\tbllab}[1]{\label{table:#1}}
\newcommand{\tblref}[1]{\HLink{Table}{table:#1}}

\newcommand{\seclab}[1]{\label{sec:#1}}
\newcommand{\secref}[1]{\HLink{Section}{sec:#1}}

\newcommand{\defrefY}[2]{\hyperref[def:#1]{#2}}

\providecommand{\eqlab}[1]{}%
\renewcommand{\eqlab}[1]{\label{equation:#1}}

\providecommand{\remove}[1]{}%
\newcommand{\Set}[2]{\left\{ #1 \;\middle\vert\; #2 \right\}}

\newcommand{\pth}[1]{\mleft(#1\mright)}%

\newcommand{\ceil}[1]{\mleft\lceil {#1} \mright\rceil}
\newcommand{\floor}[1]{\mleft\lfloor {#1} \mright\rfloor}

\newcommand{\brc}[1]{\left\{ {#1} \right\}}

\newcommand{\cardin}[1]{\left\lvert {#1} \right\rvert}%

\newlist{compactenumA}{enumerate}{5}%
\setlist[compactenumA]{itemsep=-0.5ex,topsep=0.5ex,partopsep=1ex,parsep=1ex,%
   label=(\Alph*)}%

\newlist{compactenuma}{enumerate}{5}%
\setlist[compactenuma]{itemsep=-0.5ex,topsep=0.5ex,partopsep=1ex,parsep=1ex,%
   label=(\alph*)}%

\newlist{compactenumI}{enumerate}{5}%
\setlist[compactenumI]{itemsep=-0.5ex,topsep=0.5ex,partopsep=1ex,parsep=1ex,%
   label=(\Roman*)}%

\newlist{compactenumi}{enumerate}{5}%
\setlist[compactenumi]{itemsep=-0.5ex,topsep=0.5ex,partopsep=1ex,parsep=1ex,%
   label=(\roman*)}%

\newlist{compactitem}{itemize}{5}%
\setlist[compactitem]{itemsep=-0.5ex,topsep=0.5ex,partopsep=1ex,parsep=1ex,%
   label=\ensuremath{\bullet}}%

\numberwithin{figure}{section}%
\numberwithin{table}{section}%
\numberwithin{equation}{section}%

\def\DD{{{\bf \delta}}}
\def\CH{{\mathop{\mathrm{ConvexHull}}}}
\newcommand{\inter}{\mathop{\mathrm{int}}}

\def\findEdgeDirection{{{\sc{Find\-Edge\-Direction}}}}
\def\gcdGeometric{{{{\sc{Geom\-GCD}}}}}

\def\DCH{{\mathop{\mathrm{Discrete{}Hull}}}}
\newcommand{\LL}{{\cal {L}}}
\newcommand{\D}{{\cal D}}
\newcommand{\G}{{\hat G}}
\newcommand{\T}{{\cal T}}

\newcommand{\ZZ}{\mathbb{Z}}
\def\bd{{\partial}}

\newcommand{\triangleDCHSize}{2 \ceil{\log_\phi {\DD(T)} } + 7}

\newcommand{\Bray}{\xi}%

\BibLatexMode{%
   \bibliography{dch}
}

\title{An Output Sensitive Algorithm for Discrete Convex Hulls%
   \thanks{%
      This work has been supported by a grant from the U.S.--Israeli Binational Science Foundation.  This work is part of the author's Ph.D. thesis, prepared at Tel-Aviv University under the supervision of Prof. Micha Sharir. It appeared in SoCG 98 \cite{h-osafd-98} and CGTA \cite{h-osafd-98a}.%
   }%
}

\author{%
   Sariel Har-Peled%
   \thanks{%
      School of Mathematical Sciences, Tel Aviv University, Tel Aviv 69978, Israel; {\tt sariel@math.tau.ac.il}.}%
}

\date{September 22, 1997%
   \footnote{Re\LaTeX{}ed on \today.}}

\begin{document}
\maketitle

\begin{abstract}
    Given a convex body $C$ in the plane, its discrete hull is $C^0 = \CH( C \cap \LL )$, where $\LL = \ZZ \times \ZZ$ is the integer lattice. We present an $O( |C^0| \log \DD(C) )$-time algorithm for calculating the discrete hull of $C$, where $|C^0|$ denotes the number of vertices of $C^0$, and $\DD(C)$ is the diameter of $C$.  Actually, using known combinatorial bounds, the running time of the algorithm is $O(\DD(C)^{2/3} \log{\DD(C)})$. In particular, this bound applies when $C$ is a disk.
\end{abstract}

\section[Introduction]{Introduction}

\begin{figure}[t]
    \centerline{ \includegraphics{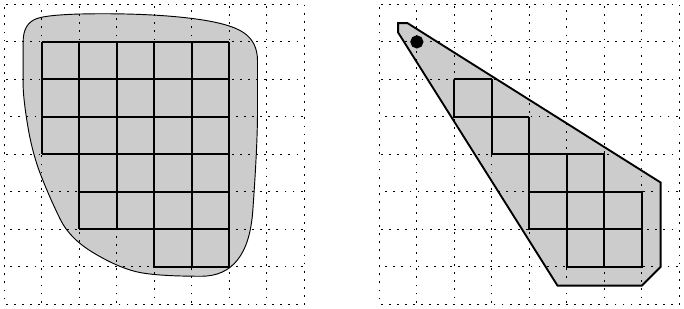} }
    \caption{On the left, a convex body $C$ such that $C \cap \LL$ is
       lattice-connected; on the right, a convex body $C$ such that
       $C \cap \LL$ is not lattice-connected}
    \figlab{grid:connected}
\end{figure}

\begin{figure}
    \centerline{ \includegraphics{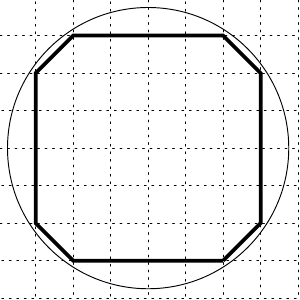} }
    \caption{A disk and its corresponding discrete convex hull}
    \figlab{discrete:hull}
\end{figure}

Let $C$ be a planar convex body which we assume to be {\em
   sufficiently round}, in the sense that the following condition
holds: Let $\LL = \ZZ \times \ZZ$ denote the planar integer lattice.
We require that $C\cap \LL$ is {\em lattice-connected}; that is, the
union of all the horizontal and vertical unit line segments that
connect between pairs of points in $C\cap \LL$ is connected; see
\figref{grid:connected}. The {\em discrete hull} $C^0$ of $C$ is
defined as the convex hull of $C \cap \LL$; see \figref{discrete:hull}
for an illustration.

The discrete hull arises in several applications. For example, in
computer graphics, $C^0$ is a natural polygonal representation of the
convex shape $C$. See \cite{BV92,EG94,KV86,Thie91} for related work.

The discrete hull has the useful property that the number of its
vertices is relatively small; more precisely, this number is
$O(\DD(C)^{2/3})$, where $\DD(C)$ is the diameter of $C$.  This is a
special case of a more general bound, established in \cite{And63}, for
convex lattice polytopes in arbitrary dimension.  For the sake of
completeness, we provide a simple proof for the planar case (see
\lemref{dch_complexity} below).  A similar proof has also been given
in \cite{KV86}.  This property makes the discrete hull an attractive
and useful tool for approximating convex shapes (see
\cite{hsv-aspcp-96,ahsv-aspcp-97} for an application of discrete hulls
in three dimensions for computing approximate shortest paths on convex
polytopes).

Solving a problem raised by Pankaj Agarwal, we present in this paper
an output-sensitive algorithm for calculating the discrete hull $C^0$
of a sufficiently round convex body $C$. The algorithm runs in time
$O( |C^0| \log \DD(C) )$, where $|C^0|$ is the number of vertices of
$C^0$.  By the above results, this bound is also
$O(\DD(C)^{2/3} \log \DD(C) )$.

Suppose we associate a value with each point of $G$, where $G$ is all
the points of the lattice inside a square of size $N\times N$. Katz
and Volper \cite{KV86} present a data-structure for an intermixed
sequence of updates and retrievals of the sum of values of the points
of $G$ within a disk, such that a retrieval or update operation takes
$O(N^{2/3}\log^3 N)$ time. (An update operation modifies the value
associated with a vertex of $G$) The first stage in answering a query,
in the algorithm of \cite{KV86}, is the computation of the discrete
hull of the query disk. Katz and Volper do not present an efficient
algorithm for computing the discrete hull of a disk, and their bound
holds when counting only the semi-group operations (see \cite{KV86}).
By our main result, their algorithm can be extended such that their
bound holds in the regular RAM model as well.

Our algorithm employs a variant of the continued-fraction technique
used to approximate a real number by rationals with small
denominators.  We first review the machinery needed for this algorithm
in \secref{rational}, and give a geometric interpretation of it that
we will later exploit. The connection between continued fractions and
the discrete hull of a half-plane was already observed by Klein in
1895 (see \cite{Dav83}). Furthermore,
\cite{GY86,lc-avcpg-92,ks-bcmlp-96} use continued fractions to
calculate convex polygonal chains that are discrete hulls of some
specific shapes. We extend their techniques to handle the more general
case of an arbitrary convex body.

Our discrete hull algorithm is presented and analyzed in
\secref{algorithm}. First, in \secref{bounds}, we derive several
combinatorial bounds on the complexity of the discrete hull. Next, we
describe, in \secref{algorithm:2}, our algorithm in detail, and present
several variants of the algorithm.

We have implemented this algorithm for the special case of a disk.
The experimental results are reported in \secref{experimental}.

\section[The Geometric Interpretation of Continued Fractions] {The
   Geometric Interpretation of Continued\\ Fractions}

\seclab{rational}

In this section we review the machinery developed for the computation
of continued fractions and discuss its geometric interpretation. For a
review of continued fractions see \cite{Dav83}.

Given a positive real number $r$, it has a unique representation as a
simple continued fractions:
\[
    r = q_0 + \frac{1}{q_1 + \frac{1}{q_2 + \frac{1}{q_3 + \cdots}}} =
    q_0 + \frac{1}{q_1+} \frac{1}{q_2+} \frac{1}{q_3+} \cdots,
\]
where $q_i > 0$, for $i > 0$, $q_0 \geq 0$, and $q_0, q_1, \ldots$ are
integers. If $r = a/b$ is a rational (in reduced form), the continued
fraction representation of $r$ is finite, namely
$r = q_0 + \frac{1}{q_1+} \cdots \frac{1}{q_{n}}$. Consider in what
follows such an $r$. Moreover, the following recurrence holds
\[
    q_i = \floor{x_i},\quad
    x_{i+1} = \frac{1}{x_i - q_i},
    \quad\text{for}\quad i=0, \ldots, n,
\]
where $x_0 = r$. In particular, the terms are byproduct of the
calculation of the $gcd(a,b)$, and thus can be calculated in
$O(\log{ \min( a, b ) })$ time (See \cite{clr-ia-90}).

Given a reduced fraction $r' = a/b$, we define $G(r' ) = (b, a)$, a
point of the integer lattice corresponding to $r'$.

The intermediate continued fractions, $r_i = q_0 + \frac{1}{q_1+} \cdots \frac{1}{q_{i}}$, for $i=0, \ldots, n$, are called the {\em convergents} of $r$. Let $\frac{a_i}{b_i}$ be the reduced form of $r_i$, for $i=0, \ldots, n$.  The convergent $r_i$ is an optimal approximation to $r$, in the sense that if $\frac{c}{d}$ is between $r$ and $r_i$ and $\cardin{ \frac{c}{d} -r } < \cardin{ \frac{a_i}{b_i} - r }$ then $ d > b_i$, for $i = 1, \ldots, n$. See \cite{Dav83}.

Let $p = G(r) = (b,a)$ be the point corresponding to $r$, and let $p_i = G( r_i )= (b_i, a_i)$ be the planar grid point corresponding to $r_i$, for $i=0, \ldots, n$.  We then have $p_i = q_i p_{i-1} + p_{i-2}$, for $i=0, \ldots, n$, where $p_{-2} =(1, 0)$ and $p_{-1} = (0, 1)$.  Furthermore, the only points of $\LL$ in $\triangle{op_{i}p_{i+1}}$ are its vertices, for $i=0, \ldots, n-1$ (see \cite{Dav83}). We refer to $p_{0}, \ldots, p_n$ as the (geometric) convergents of $p$. All this is illustrated in \tblref{gcd}.

\begin{table}
    \centering%
\centering
\begin{tabular}{||r|r||l|l||r||}
\hline
~ & ~ & ~ &
Representation  &
 The \hspace{1.5cm}
~\\
{$a$}  &
 {$b$}  & {$\floor{{a}/{b}}$} &
{of $ ( 14, 31 ) $} &
{ convergent $ r_i $ } %
\\\hline
$ 31 $  &  $ 14 $  &  $ 2 $ & $ 31 (0, 1) + 14 (1, 0) $ &
        $2\bigr.$
\\ %
$ 14 $  &  $ 3 $  &  $ 4 $ & $ 14 (1, 2) + 3 (0, 1) $ &
        $2 + \frac{1}{4\bigr.}$
\\ %
$ 3 $  &  $ 2 $  &  $ 1 $ & $ 3 (4, 9) + 2 (1, 2) $ &
        $2 + \frac{1}{4 + \frac{1}{1\bigr.}}$
\\ %
$ 2 $  &  $ 1 $  &  $ 2 $ & $ 2 (5, 11) + 1 (4, 9) $ &
        $2 + \frac{1}{4 + \frac{1}{1 + \frac{1}{2\bigr.}}}$
\\ %
\hline
\end{tabular}

    \caption{The convergentsrgents of the number
   $\frac{31}{14}$ can be obtained by calculating $gcd(31,14)$ by
   Euclid algorithm.}
\tbllab{gcd}
\end{table}

Given two vectors $u, u' \in \LL$, where the coordinates of $u'$ are
relatively prime, let $r = ray( u, u')$ be the ray emanating from $u$
in the direction of $u'$. Let $p( r, i) = u + i u'$ be the $i$-th
lattice point of the ray, for $i \geq 0$. We denote by $r[i]$ the
closed segment on $r$ between $p(r,i)$ and $p(r, i+1)$, for
$i \geq 0$. We denote by $\Bray( u, u')$ the {\em discrete ray}
emanating from $u$ in the direction $u'$:
\[
    \Bray( u, u' ) = \Set{ u + i u' }{ i \in \ZZ, i > 0 }.
\]

The execution of the $gcd$ algorithm can be interpreted as a geometric algorithm in the following manner: The intermediate numbers $a_i,b_i$ in the $i$-th stage of the $gcd$ algorithm are the coefficients of $p = G(r)$ in its representation in the base $\brc{ p_{i-1}, p_{i-2} }$.  Each iteration refines the base by replacing its shorter vector by a longer one (the new convergent), such that $p$ remains `in the middle' (i.e., $p$ remains a positive combination of the two base vectors).  This process continues until the last convergent is equal to $p$.

\begin{figure}[tbp]
    \begin{center}
        \begin{minipage}[c]{5.0in}

            \noindent {\large{\bf{Algorithm:}}}\ \ \ {\gcdGeometric}(
            $r$ )

            \noindent \ \ {\tt Input:} A rational number $r$.\

            \noindent \ \ {\tt Output:} The convergents of $r$:
            $p_0, \ldots, p_n$.

            \begin{itemize}
                \item[(I)] Set
                $p_{-2} \leftarrow (1, 0), p_{-1} \leftarrow (0, 1)$,
                $i \leftarrow 0$.

                Let $l$ be the line $y = r x$.

                \item[(II)] Let $\Bray_i = \Bray( p_{i-2}, p_{i-1})$. If
                $\Bray_i \cap l \ne \emptyset$ then let $q_i$ be the
                index of the point lying on $l$, namely
                $p_i = q_i p_{i-1} + p_{i-2} \in l$. Clearly, $p_i$ is
                the last convergent of $r$ and we are done.

                Otherwise, let $q_i$ be the largest integer such that
                $p_i = q_i p_{i-1} + p_{i-2}$ and $p_{i-1}$ lie on
                different sides of $l$.

                \item[(III)] $i \leftarrow i + 1$. Goto step (II).
            \end{itemize}

        \end{minipage}
        \begin{minipage}[c]{6.0in}
            \caption{The algorithm for computing the discrete-hull
               simulates the above algorithm, that calculates the
               convergents of a rational number in a geometric
               manner.}%
            \figlab{geometric:gcd}
        \end{minipage}
    \end{center}

\end{figure}

Thus, calculating the convergents of $r$ can be done in a geometric
setting. The algorithm {\gcdGeometric} that does so is presented in
\figref{geometric:gcd}, and an example of its performance is shown in
\figref{gcd:run}. Note that in the algorithm, the division operation
of the $gcd$ is replaced by a discrete ray-shooting.

If $r > 0$ is irrational, the odd convergents $p_{-1}, p_1, \ldots$
are the vertices of the discrete hull of the region that lies above
the line $l$ in the positive quadrant of the plane, where $l$ is the
line $y = r x$. Similarly, the even convergents $p_{-2}, p_0, \ldots$
are the vertices of the discrete hull of the region that lies below
$l$ in the $x$ positive half-plane. Those two polygonal chains lie in
the half-planes defined by $l$, and they intersect $l$ only if $r$ is
rational (see \cite{Dav83}).

\begin{figure}
    \centering
    \begin{minipage}[c]{3.6in}
        \vspace{-8cm}
        \begin{tabular}{||r|r||l|l||l||}
          \hline
          ~
          & ~
          & ~
          &
            Current
          &
            The  %
                                                           ~\\
          ~
          &
            ~
          &
            ~
          &
            representation & convergent\\
          {$a$}
          &
            {$b$}
          & {$\floor{{a}/{b}}$}
          &
            {of $ ( 5, 8 ) $ } & {$ r_i $ }
          \\\hline
          $ 8 $
            &  $ 5 $
                &  $ 1 $

          &
            $ 8 (0, 1) + 5 (1, 0) $
       &                                                                $1\bigr.$
          \\ %
          $ 5 $  &  $ 3 $  &  $ 1 $ & $ 5 (1, 1) + 3 (0, 1) $ &
                                                                $1 + \frac{1}{1\bigr.}$
          \\ %
          $ 3 $  &  $ 2 $  &  $ 1 $ & $ 3 (1, 2) + 2 (1, 1) $ &
                                                                $1 + \frac{1}{1 + \frac{1}{1\bigr.}}$
          \\ %
          $ 2 $  &  $ 1 $  &  $ 2 $ & $ 2 (2, 3) + 1 (1, 2) $ &
                                                                $1 + \frac{1}{1 + \frac{1}{1 + \frac{1}{2\bigr.}}}$
          \\ %
          \hline
        \end{tabular}
    \end{minipage}
    \includegraphics{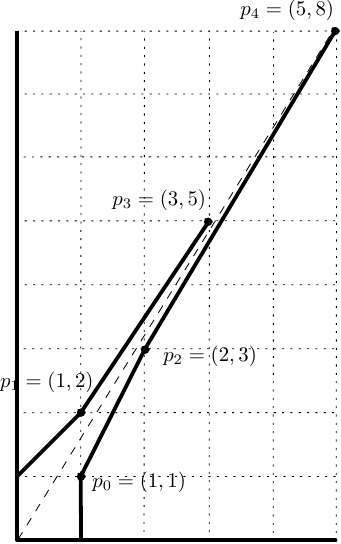}

    \caption{The convergents
       generated by \gcdGeometric( $8/5$ )}
    \figlab{gcd:run}
\end{figure}

\section[The Discrete Hull]{The Discrete Hull}

\seclab{algorithm}

In this section, we introduce the notion of discrete hull and provide
an algorithm for its computation. In \secref{bounds}, we prove several
bounds on the complexity of the discrete hull. In \secref{algorithm:2},
we describe algorithm for computing the discrete hull.

\subsection[Combinatorial Bounds on the Complexity of the Discrete Hull]
{Combinatorial Bounds on the Complexity of the Discrete Hull}
\seclab{bounds}

Given a convex body $C$ in the plane, we define the {\em discrete
   hull} $C^0$ of $C$ as the convex hull of $C \cap \LL$, where
$\LL = \ZZ \times \ZZ$ is the integer lattice.  The following lemma is
well known \cite{And63,BB91,KV86,Thie91}. For the sake of completeness
we provide here a simple proof. A similar proof has been given in
\cite{KV86}.

\begin{lemma}
    \lemlab{dch_complexity}%
    Let $C$ be a convex polygon in the plane, all of whose vertices belong to $\LL$, and let $D$ be its diameter. Then $C$ has $O( D^{2/3} )$ vertices.
\end{lemma}

\begin{proof}
    The polygon $C$ is contained in an axis-parallel square of size
    $D$, thus the perimeter $P(C)$ of $C$ is at most $4D$.

    Let $p_1, \ldots, p_m$ be the vertices of $C$ in their
    counterclockwise order around $C$.

    Let $V = \brc{ v_1, \ldots, v_m } $ be the set of vectors
    connecting consecutive vertices on $\bd{C}$; that is,
    $v_i = p_{i+1} - p_i$, for $i=1,\ldots,m-1$, and
    $v_m = p_1 - p_m$.

    Let $l$ be a parameter to be specified shortly.  By definition,
    $P(C) = \sum_{v \in V} |v|$. Thus, there are at most $4D/l$
    vectors in $V$ of length larger than or equal to $l$.

    Clearly, $V$ cannot contain three collinear vectors, for otherwise
    one of the points $p_i$ would have to lie in the relative interior
    of an edge of $\bd{C}$, so it cannot be a vertex of $C$. This
    argument also implies that any pair of collinear vectors in $V$
    must have opposite orientations.  Thus the number of vectors in
    $V$ of length smaller than $l$ is at most the number of lattice
    points at distance smaller than $l$ from the origin, which in turn
    is at most $4 l^2$.

    Thus the number of vertices of $C$ is upper bounded by
    $\frac{4D}{l} + 4l^2$. Set $l = (D/2)^{1/3}$, the bound becomes
    $6 \cdot \sqrt[3]{2} D^{2/3}$ and the lemma follows.
\end{proof}

For the a case of disk, we have the following extension of \lemref{dch_complexity}.

\begin{lemma}
    Let $D$ be a disk of radius $r > 1$ in the plane. Then
    $|\LL \cap \bd{\DCH(D)}| = O(r^{2/3})$.

    \lemlab{dhs:complexity}
\end{lemma}

\begin{proof}
    Let $D^0 = \DCH(D)$.  For an edge $e = \overrightarrow{pq}$ of
    $D^0$, we define its {\em weight} to be $w(e) = |\LL \cap e| - 1$.
    We assume that $e$ is oriented so that $\inter{D^0}$ lies to the
    left of $e$.  Let $v=v(e)$ be the direction vector of $e$; namely,
    the shortest lattice vector which is parallel to $e$ and points in
    the same direction. Let $c=||v||$. Clearly, $w(e) = ||e||/c$.

    Let $l$ be the line passing through $e$, and let $H^-$ denote the
    half-plane bounded by $l$ that does not contain the center of $D$.
    Let $D'$ denote the disk of radius $r$ whose boundary passes
    through $p$ and $q$, and whose center lies in the same side of $l$
    as the center of $D$.

    Clearly, $H^- \cap D' \subseteq H^- \cap D$. Let $s$ denote the
    point on $H^- \cap \bd{D'}$ with a tangent to $D'$ that is
    parallel to $e$. Let $h$ be the distance between $s$ and $e$. See
    \figref{circle}.

    \begin{figure}
        \centerline{ \includegraphics{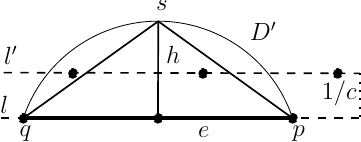} }
        \caption{Illustrating the proof that $h \leq 2/c$ when
           $w(e) > 1$.  }
        \figlab{circle}
    \end{figure}

    We claim that if $w = w(e) > 1$ then $h \leq 2/c$. Indeed, let
    $l'$ be the translation of $l$ by distance $1/c$ into $H^-$. The
    line $l'$ contains points of the lattice $\LL$, and the distance
    between two consecutive lattice points on $l'$ is $c$ (See
    \cite[Chap.  3]{hw-tn-65}). Thus, if $h \geq 2/c$ then
    $||l' \cap D|| \geq ||l' \cap D'|| \geq c$. This however imply
    that there must be a point of $\LL$ inside $H^- \cap D'$,
    contradicting the fact that $e$ is an edge of $D^0$.

    An easy calculation shows that $h=r- \sqrt{r^2 - w^2c^2/4}$. Thus,
    if $w > 1$ then $r- \sqrt{r^2 - w^2c^2/4} \leq 2/c$, implying that
    \[
        \frac{w^2c^2}{2r} \leq \frac{w^2c^2}{r+ \sqrt{r^2 -
              \frac{w^2c^2}{4}} } \leq \frac{8}{c}.
    \]
    Thus $w(e) = w \leq \frac{4\sqrt{r}}{c^{3/2}}$. This implies that,
    for $w \geq 2$, the maximum possible value of $c$ is
    $c \leq {\ceil{\sqrt[3]{4r}}}$.

    The number of lattice vectors of length between $c$ and $c+2$ is $O(c)$, and, by \lemref{dch_complexity}, the number of edges of $D^0$ is $O(r^{2/3})$. Hence, the number of lattice points lying on the edges of $D^0$ is bounded by the number of edges of weight one, plus the number of lattice points lying on edges of weight larger than 1; that is, this number is
    \begin{equation*}
        O\pth{r^{2/3}} + \sum_{c=1}^{\ceil{\sqrt[3]{4r}}} O(c)\cdot
        \frac{4\sqrt{r}}{c^{3/2}}
        =%
        O\bigl( r^{2/3}\bigr)
        +
        O \Bigl( \sqrt{r}
        \cdot \sum_{c=1}^{\ceil{\sqrt[3]{4r}}} \frac{1}{\sqrt{c}}
        \Bigr)
        = O
        \pth{ r^{2/3}}.
    \end{equation*}
\end{proof}

\begin{remark}
    \lemref{dhs:complexity} was proved independently by B{\'a}r{\'a}ny
    \cite{bar-pc-96} (but was not published).  Furthermore,
    B{\'a}r{\'a}ny and Larman \cite{bl-chip-96} showed that the number
    of $k$-dimensional faces of the discrete hull of the ball of
    radius $r$ centered at the origin (in dimension $d$) is
    $\Theta ( r^{\frac {d(d-1)}{d+1}} )$, for $k=0,\ldots,d-1$.  All
    these results were obtained recently and independently after the
    original preparation of this paper.
\end{remark}

\begin{remark}
    For a compact convex and smooth shape $C$ in the plane, such that
    the curvature of $C$ is bounded from below by $c >0$, the number
    points $|\LL \cap \bd{\DCH(C)}|$ is $O((1/c)^{2/3})$. This follows
    by a straightforward extension of the proof of
    \lemref{dhs:complexity}.
\end{remark}

\begin{lemma}
    Given a triangle $T$, such that two of its vertices belong to
    $\LL$, the number of vertices of the discrete hull of $T$ is at
    most $\triangleDCHSize$, where $\phi = \frac{\sqrt{5} + 1}{2}$ is
    the Golden Ratio.

    \lemlab{idiotic:above}
\end{lemma}

\begin{proof}
    Kahan and Snoeyink showed in \cite{ks-bcmlp-96} how to compute the
    discrete hull inside such a triangle in $O(\log \DD(T))$ time.  A
    careful inspection of their algorithm reveals that it generates at
    most $\triangleDCHSize$ vertices.
\end{proof}

\begin{lemma}
    Given a polygon $P$ with $n$ edges in the plane, the complexity of
    the discrete hull of $P$ is at most $\pth{ \triangleDCHSize} n$,
    where $\DD(P)$ is the diameter of $P$.

    \lemlab{complexity:hull:polygon}
\end{lemma}

\begin{proof}
    Let $P^0$ be the discrete hull of $P$, and $p$ a vertex of $P$.
    Let $r^+,r^-$ denote the two rays that emanate from $p$ and are
    tangent to $P^0$, and let $v_p^+, v_p^-$ be the two vertices of
    $P^0$, furthest from $p$, that lie on $r^+, r^-$, respectively.
    Clearly, the vertices of $P^0$ inside $T= \triangle{pv^-v^+}$
    belong to the discrete hull of $T$. By \lemref{idiotic:above}, the
    number of those vertices (including $v^-$ and $v^+$), which, in
    particular, include all vertices of $P^0$ that are visible from
    $p$, is at most $\triangleDCHSize$.

    Applying the above argument to all the vertices of $P$, we have
    that the number of vertices of $P^0$ is bounded by
    $\pth{ \triangleDCHSize }n$, since each vertex of $P^0$ is visible
    from at least one of the vertices of $P$.
\end{proof}

\subsection{The Discrete Hull Algorithm}
\seclab{algorithm:2}

In this section, we present an output sensitive algorithm for
computing the discrete hull.  The main procedure that we present here
receives as input a sufficiently round convex body $C$, and a vertex
$p$ of $C^0$.  It computes the next vertex $p'$ of $C^0$
counterclockwise from $p$.  The full algorithm first computes an
initial (e.g., the lowest) vertex $p_0$ of $C^0$, and then applies
repeatedly the above procedure until all vertices of $C^0$ are
obtained.

We assume a model of computation in which $C$ is given implicitly so
that (a) we are given a lattice point inside $C$, and (b) we can
obtain, in constant time, the intersection points of any query line
with $\bd{C}$, when these points exist. Moreover, we also include the
floor function in in our model of computation. This allows us also to
perform, in constant time, discrete ray shootings, where, for a query
discrete ray $\Bray(u, v)$, we want to find its first point inside or
outside $C$.

Let $T = \triangle{ovw}$ be a right-angle triangle, such that $o$ is
the origin, $v = (x(v), y(v))$ lies in the positive quadrant of the
plane, and $w = (x(v), 0)$ lies on the $x$-axis. The point $u$ of the
(upper) edge $ou$ of the discrete hull of $T$ is a multiple of $u'$,
where $u'$ is the largest (even) convergent of $v$ that still lies
inside $T$. (See \figref{gcd:run} for an illustration: Let $v$ be a
point slightly below $p_4$; then $op_2$ is the upper edge of the
discrete hull.) This observation (or a variant of it) has been used in
calculating parts of the discrete hull in some special cases, where
the body is polygonal or given in an explicit form.  See
\cite{GY86,lc-avcpg-92,ks-bcmlp-96}. In fact, $u$ is calculated in
\cite{ks-bcmlp-96} by performing ray-shootings inside and outside $T$,
an idea that we extend to handle our more general settings.

Unfortunately, these methods fail when the body is either given in an
implicit form or is not polygonal. We present an extension to the
above techniques, where we compute the convergents of the required
edge direction by performing a sequence of ray-shootings on the given
body.

\begin{lemma}
    Let $\rho=(a,b)$ be a point of the lattice $\LL$, such that
    $a, b >0$ and $gcd( a, b ) = 1$. Let
    $\rho_0, \ldots, \rho_n = \rho$ be the sequence of convergents of
    $\rho$.

    Then the segment $s = \rho_i \rho_i'$ must intersect the segment
    $op$, where $o$ denote the origin and
    $\rho_i' = \rho_i + \rho_{i-1}$, for $i=0, \ldots, n$.

    \lemlab{intersect}
\end{lemma}

\begin{proof}
    Let $l$ be the line passing through $o$ and $\rho$. By definition,
    $\rho_i = q_i \rho_{i-1} + \rho_{i-2}$ and $q_i$ is maximal such
    that $\rho_i$ lies on the same side of $l$ as $\rho_{i-2}$. Thus,
    $\rho_i \rho_i'$ must intersect $l$, or $q_i$ would not be
    maximal.

    Moreover, $\rho_i \rho_i'$ must intersect $o \rho$, or else the
    vector $\rho_{i+1}$ would be longer than $\rho$, which is
    impossible.~
\end{proof}

Let $e'$ be the (currently unknown) vector connecting $p$ with $p'$.
Let $e$ be the reduced vector (i.e. $gcd(x(e), y(e) ) = 1$) having the
same direction as $e'$.  Given $e$, we can easily obtain $p'$ by
calculating the last intersection between $C$ and the discrete ray-$\Bray(p,e)$. Without loss of generality, we assume that both
coordinates of $e$ are positive. This issue is discussed later in more
detail.

In the following, we assume, for simplicity, that $p$ is the origin.
The algorithm, {\findEdgeDirection}, calculates the convergents of
$e$, and thus $e$ itself, by simulating the execution of
{\gcdGeometric} on $e$. The idea is to replace the half-plane
ray-shooting of {\gcdGeometric} by ray shootings on $\bd{C}$.
{\findEdgeDirection} generates several candidates for the vector $e$,
because, unlike {\gcdGeometric}, the vector $e$ is not given
explicitly, and the algorithm does not know when to stop.  The
required vector $e$ is simply the clockwise-most vector out of the
group of generated candidates.

Let $l$ be the line passing through $e$. Denote by $H^{+}$ the open
half-plane lying to the left of $e$, and by $H^-$ the open half-plane
lying to the right of $e$.

The odd convergents of $e$ are the approximations to $e$ from `inside'
$C$ (namely, these convergents lie inside $H^+$ and may lie inside
$C$), and the even convergents are the approximations to $e$ from the
`outside' of $C$ (namely, these convergents do not lie inside $C$, and
lie inside $H^-$, except maybe for the last convergent that may lie on
$l$).

We now describe {\findEdgeDirection} in detail. The algorithm starts with $p_{-2} =(1, 0)$ and $p_{-1} = (0,1)$. In the $i$-th stage the algorithm calculates the intersection between $ray_i = ray( p_{i-2}, p_{i-1} )$ and $C$. By the relation $p_i = q_i p_{i-1} + p_{i-2}$ it follows that the next convergent of $e$ lies on the discrete ray $\Bray_i = \Bray( p_{i-2}, p_{i-1} )$. If $ray_i$ does not intersect $C$ then we stop the procedure, since by \lemref{intersect} we know that all the convergents of $e$ were computed.  Otherwise, we apply \lemref{intersect} to decide what is the next convergent. There are two cases:

\begin{itemize}
    \item {\bf $i$ is odd:} Let $k$ be maximal such that
    $ray_i[k] \cap C \ne \emptyset$. If $k = 0$ then we are done, in
    the sense that all convergents of $e$ have already been computed
    (otherwise, we would have $p_i = p_{i-2}$, which is impossible,
    for $i > 0$, because all the convergents are distinct).  Set
    $p_i \leftarrow p(ray_i, k)$. If $p_i \in C$ then $p_i$ might be
    $e$, and thus it is a candidate for the vector $e$.

    \item {\bf $i$ is even:} Let $k$ be minimal such that
    $ray_i[k] \cap C \ne \emptyset$. If $p(ray_i, k+1) \in C$, then it
    is a candidate for the vector $e$, otherwise $p(ray_i, k)$ is a
    convergent of $e$. If $k=0$ and $i > 0$ then we are
    done. Otherwise, set $p_{i} \leftarrow p(ray_i, k )$.
\end{itemize}

The algorithm returns the clockwise-most vector out of the candidates
generated. See \figref{find:edge:iter}.

\begin{figure}
    \centerline{ \includegraphics{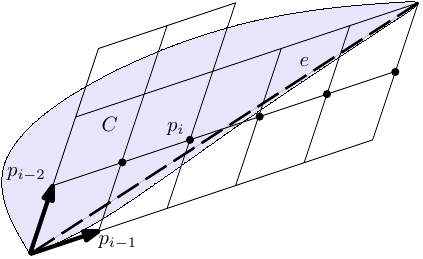} }
    \caption{Illustrating an iteration of {\findEdgeDirection}}
    \figlab{find:edge:iter}
\end{figure}

Let $p_0, \ldots, p_n$ be the convergents of $e$.  The correctness of
{\findEdgeDirection} follows from the fact that it performs exactly
the same ray-shootings as {\gcdGeometric}($e$) performs, up to the
$(n-1)$-th stage, and selects the correct points on these rays to be
the convergents of $e$.  It is easy to verify, that if one of the
results of these $n-1$ ray-shootings of {\findEdgeDirection} differs
from the corresponding $(n-1)$ ray-shootings of {\gcdGeometric}($e$),
then $e$ is not the direction of the next edge of the discrete hull,
because there exists a vector of $\LL$ that lies inside $C$ and is
clockwise to $e$. This implies that $e$ is not an edge of the discrete
hull, which is a contradiction.

As for the running time of the algorithm, we notice that the
coordinates of the sequence of convergents that {\findEdgeDirection}
generates, are a ``super-Fibonacci'' sequence since
$x(p_i) \geq x(p_{i-1}) + x( p_{i-2})$ and
$y(p_i) \geq y(p_{i-1}) + y( p_{i-2})$. Thus
$||p_i|| \geq |F_i| = \floor{ \phi^k/ \sqrt{5} }$, where $F_i$ is the
$i$-th Fibonacci number and $\phi = \frac{\sqrt{5} + 1}{2}$ is the
Golden Ratio (See \cite{clr-ia-90}).

Thus, the $\pth{ 2 + \ceil{\log_{\phi} \DD(C) }}$-th convergent
generated by {\findEdgeDirection} must lie outside $C$, and the ray
created in the next step will not intersect $C$, implying that the
algorithm will have terminated by this iteration. Each iteration of
the algorithm takes $O(1)$ time, and, as just argued, there are at
most $O(\log \DD(C) )$ iterations. Thus, the total running time of
{\findEdgeDirection} is $O( \log( \DD(C) ))$.  However, it might be
that no candidate was generated by {\findEdgeDirection}.  This implies
that our assumption that $x(e), y(e) \geq 0$ is false, and $e$ does
not lie in the positive quadrant of the plane.

Thus, during the calculation of the discrete hull, the algorithm can
either maintain the current quadrant of the plane containing $e$, or
alternatively, perform a search for $e$ in each quadrant of the plane.
Picking $e$ out of the (at most) four candidates generated can be done
in $O(1)$ time, by computing their convex hull $C_H$, where $e$ is the
first edge of $C_H$ encountered when tracing $\bd{C_H}$ in a
counterclockwise direction starting from $o$.

We thus obtain:

\begin{theorem}
    \thmlab{discrete:ch:almost}%
    Given (i) a compact convex shape $C$ in the plane such that the
    intersection between $C$ and a line can be calculated in $O(1)$
    time, and (ii) a vertex $v_0$ of the discrete hull $C^0$ of $C$,
    then $C^0$ can be calculated in $O( |C^0| \log{\DD(C)})$ time and
    $O(1)$ space, where $|C^0|$ denotes the number of vertices of
    $C^0$.
\end{theorem}

Although the assumption in \thmref{discrete:ch:almost} on the
availability of a vertex of the discrete hull seems to be
self-defeating, it can be replaced by a more natural condition, as
follows.

\begin{defn}[Lattice Connected Set]
    Let $\G = (\LL, E)$ be the {\em lattice graph}, whose edges
    connect every pair of lattice points at distance 1.  A subset $S$
    of $\LL$ is {\em lattice-connected} if the induced subgraph $\G_S$
    is connected.
\end{defn}

\begin{lemma}
    \lemlab{bounding:box}%
    Given a compact convex body $C$ in the plane, such that
    $C \cap \LL$ is lattice-connected, and a point $v \in C \cap \LL$,
    then a point on the discrete hull of $C$ can be calculated in
    $O(\log{\DD(C)})$ time.
\end{lemma}

\begin{proof}
    The idea is to perform an unbounded search for the lowest vertex
    of the discrete hull of $C$. We intersect $C$ with the horizontal
    lines $l_i$ passing through $v+(0,-2^i)$, for $i=1,2,\ldots$, and
    stop as soon as we encounter a horizontal line, $l_k$, that does
    not intersect $C \cap \LL$.  We then perform a binary search for
    the bottom horizontal line of the grid that intersects $C$,
    between the horizontal line passing through $v$ and $l_k$. The
    lattice-connectivity of $C\cap\LL$ guarantees the correctness of
    this procedure.

    Since $k \leq 1 \ceil{\log \DD( C ) }$, it follows that the binary
    search is performed on an interval of size $O( \DD( C ) )$.

    Given this line, we can calculate a vertex of the discrete hull by
    performing a single ray-shooting.~
\end{proof}

We thus obtain the main result of this paper:

\begin{theorem}
    Given a compact convex shape $C$ in the plane such that the
    intersection between $C$ and any line can be calculated in $O(1)$
    time, such that $C\cap \LL$ is lattice-connected and such that a
    point in $C \cap \LL$ is available, then the discrete hull $C^0$
    of $C$ can be calculated in $O( |C^0| \log{\DD(C)})$ time and
    $O(1)$ space, where $|C^0|$ denotes the number of vertices of
    $C^0$.

    \thmlab{discrete:ch}
\end{theorem}

In particular, we have the following corollary:
\begin{corollary}
    The discrete hull of a disk $D$ of radius $r$ in the plane can be
    calculated in $O(r^{2/3} \log{r} )$ time and $O(1)$ space.
\end{corollary}

\begin{proof}
    In $O(1)$ time we can find a point $v^0 \in D \cap \LL$, by
    finding the closest point in the lattice to the center of $D$.
    Unless $r < 1/\sqrt{2}$, $D\cap \LL$ is not empty and is always
    lattice-connected, so, by \lemref{bounding:box}, a vertex of the
    discrete hull of $D$ can be calculated in $O(\log{r})$ time.  The
    corollary is now an immediate consequence of \thmref{discrete:ch}
    and \lemref{dch_complexity}.
\end{proof}

Given a convex polygon $P$ with $n$ edges in the plane, the algorithm
of Lee and Chang \cite{lc-avcpg-92} computes, in
$O( n +\log{\DD( P ) } )$ time, a vertex of the discrete hull.  All
the algorithms described so far, assume that ray-shooting on $C$ can
be computed in $O(1)$ time. In general, if $T(n)$ is the time for
answering ray-shooting queries on $C$, then the running time increases
by a factor of $T(n)$. For example, we have the following:

\begin{corollary}
    \corlab{poly:dch:alg:simp}%
    Let $P$ be a convex polygon with $n$ edges in the plane. The
    discrete hull $P^0$ of $P$ can be computed in
    $O( n + |P^0|\log{\DD(P)}\log{n} )$ time, where $|P^0|$ denotes
    the number of vertices of $P^0$.
\end{corollary}

\begin{proof}
    Using the algorithm of \cite{lc-avcpg-92}, we compute in
    $O(n + \log{\DD(P)})$ time a vertex $v_0$ of the discrete hull. We
    preprocess $P$ in $O(n)$ time to answer ray-shooting queries, in
    $O(\log{n})$ time, by computing the Dobkin-Kirkpatrick
    hierarchical decomposition of $P$ \cite{dk-ladsc-85}.

    We then compute the discrete hull using the algorithm of
    \thmref{discrete:ch}. Since each ray-shooting takes $O(\log{n})$
    time, the resulting algorithm runs in
    $O(|P^0|\log{\DD(P)}\log{n})$ time.
\end{proof}

\begin{corollary}
    The discrete hull of a convex $n$-gon $P$ can be computed in time \\
    $O(n \log{n} \log^2 \DD(P) )$.
\end{corollary}

\begin{proof}
    Readily follows from \corref{poly:dch:alg:simp} and
    \lemref{complexity:hull:polygon}.
\end{proof}

Using \lemref{complexity:hull:polygon}, it is possible to replace the
conditions in \thmref{discrete:ch} by a more natural condition.

\begin{theorem}
    Given a compact convex shape $C$ in the plane and a point
    $p \in C$, such that the intersection between $C$ and any line can
    be calculated in $O(1)$ time, the discrete hull $C^0$ of $C$ can
    be calculated in $O( \log^2{\DD(C)} + |C^0| \log{\DD(C)})$ time
    and $O(1)$ space, where $|C^0|$ denotes the number of vertices of
    $C^0$.

    \thmlab{discrete:ch:ext}
\end{theorem}

\begin{proof}
    Let $B(C)$ denote the axis-parallel bounding box of $C$, and let
    $B$ denote the discrete hull of $B(C)$ (which is an axis parallel
    rectangle). Using unbounded search (starting from a lattice point
    near $p$), similar to the one in the proof of
    \lemref{bounding:box}, one can compute $B$, in $O(\log{\DD(C)})$
    time, by performing a sequence of vertical and horizontal
    ray-shootings on $C$.

    Let $C' = C \cap B$. Clearly $\DCH(C) = \DCH(C')$, and one can
    compute the intersection of $C'$ with a line, in $O(1)$ time, by
    computing the corresponding line-intersections with $C$ and $B$.

    Let $p_0, p_1$ be two $x$-extreme points of $C'$, such that $p_0$
    and $p_1$ lie on different sides of $B$, and let $e=p_0p_1$.  If
    $|e| < 100$ then compute the discrete hull of $C$, in $O(1)$ time,
    by performing at most $2 \cdot 100$ vertical ray-shootings.

    Otherwise, let $\T$ be the vertical trapezoid such that $e$ is its
    upper edge, the length of its shorter vertical edges is $\delta$,
    and its bottom edge is horizontal, where $\delta$ is the diameter
    of $B$. See \figref{trapezoid}.

    Using the algorithm of \thmref{discrete:ch:almost}, one can
    compute, in $O(\log^2 \DD(C))$ time, a chain $S$ of $m$ adjacent
    vertices of the discrete hull of $C' \cup \T$, starting from the
    lowest lattice point lying on the left edge of $\T$, where
    $m = 1 + 4 \pth{2 \ceil{\log_\phi (2 \DD(C)) } + 7}$. If
    $C' \setminus \inter \T$ contains a vertex of $C^0$, then $S$ must
    contain such a vertex.  Otherwise, by
    \lemref{complexity:hull:polygon}, we have
    $S = \DCH(\T ) = \DCH( \T \cup C' )$.

    \begin{figure}
        \centerline{ \includegraphics{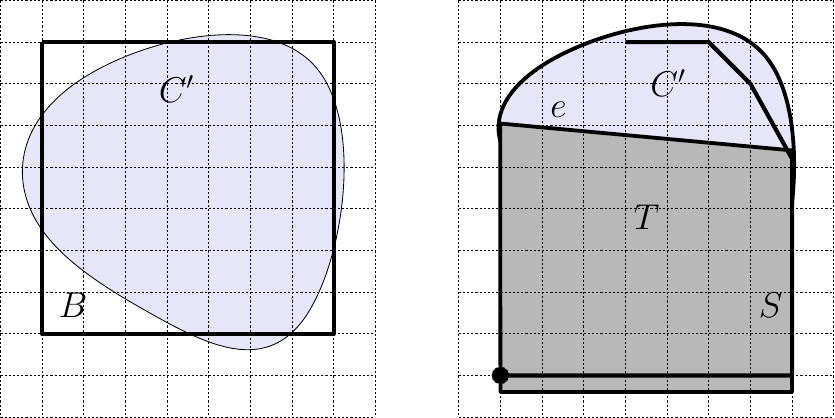} }
        \caption{Computing a vertex of the discrete hull of $C$.}

        \figlab{trapezoid}
    \end{figure}

    If no such vertex was found, apply the above procedure to the
    corresponding symmetric vertical trapezoid having $e$ as its lower
    edge, in order to compute a vertex of $C^0$ below $e$ (if such a
    vertex exists).  Thus, we can compute a vertex of $C^0$, in
    $O( \log^2 {\DD(C)} )$ time, if such a vertex exists.

    Starting from this vertex of the discrete hull, one can compute
    $C^0$, in $O \pth{ |C^0| \log{\DD(C)} }$ time, by the algorithm of
    \thmref{discrete:ch:almost}.
\end{proof}

\section[Experimental Results]{Experimental Results}
\seclab{experimental}

\begin{table}
       \centering
\centering
\begin{tabular}{|r||r|r|r|r|}
\cline{2-5}
\multicolumn{1}{c||}{} & \multicolumn{1}{c|}{Avg. Number of Points} &
 \multicolumn{1}{c|}{Max. Number of} &
 \multicolumn{2}{c|}{Time in Seconds} \\
\cline{1-1}\cline{4-5}$r$ & \multicolumn{1}{c|}{on Discrete Hull} &
 \multicolumn{1}{c|}{Iterations Per Edge}
 & DCH & IG \\
\hline
\hline
10 & 16 & 5 & 0.0015 & 0.0010 \\
100 & 74 & 6 & 0.0097 & 0.0097 \\
1,000 & 345 & 7 & 0.0550 & 0.0890 \\
10,000 & 1,603 & 9 & 0.2935 & 0.8991 \\
100,000 & 7,442 & 11 & 1.5259 & 9.0403 \\
1,000,000 & 34,535 & 13 & 7.9662 & 90.4785 \\
10,000,000 & 160,317 & 14 & 43.0855 & 906.9861 \\
\hline
\end{tabular}

       \caption{Experimental results for the discrete hull of a disk $D$ of radius $r$. The center of the disk was randomly chosen in the unit square, one hundred times for each indicated radius.  }
       \tbllab{results:rt}
   \end{table}

We have implemented the discrete hull algorithm of \secref{algorithm},
for the special case where $C$ is a disk.  The algorithm was
implemented in C++ on a Sun Sparc-10.  The implementation is
straightforward, and consists of about 150 lines of source code.

As a competing, naive algorithm, we implemented a variant of the Graham convex-hull algorithm (See \cite{o-cgc-98}), whose input consists of $O(r)$ lattice points, namely, the two extreme points on every vertical lattice line that intersects $C$. This algorithm generates the points ``on demand'', instead of calculating them in advance (thus using only $O(1)$ storage). This algorithm takes $O(r)$ time to calculate the discrete hull of a disk of radius $r$.

See \tblref{results:rt}, \tblref{results_dh}, and
\tblref{results:dhs} for results from experimenting with the
algorithm, where $IG$ and $DCH$ denote the naive and the new
algorithm, respectively.  As \tblref{results:rt} testifies, the
new algorithm is efficient in practice, and is faster than the `naive'
one, for $r \geq 100$.

As the results of \tblref{results_dh} indicate, the average number of
vertices of the discrete hull of a disk of radius $r$ is about
$3.45r^{2/3}$.  The constant $3.45$ is within the range $0.33...5.54$
proven in \cite{BB91}.  Moreover, Balog also performed experimental
testing for the case of disks centered at the origin, and got similar
results \cite{bar-pc-96}.  Thus, our results suggest that
$\lim_{r\rightarrow \infty} |DH(D_r)| / r^{2/3}$ exists and lies in
the range $3.44...3.46$, where $D_r$ is the disk of radius $r$
centered at the origin. The existence of this limit was posed as an
open question in \cite{BB91}.

We are not aware of any work that gives a rigorous analysis of this
constant, and we leave it as an open problem for further research.
Nevertheless, Balog and Deshoullier proved the existence of
$\lim_{r\rightarrow \infty} |N(r)| / r^{2/3}$, where $N(r)$ is the
average of $|DH(D_r)|$ over a small interval $[r, r+ H]$ for some
$H<1$ \cite{bar-pc-96}.

\begin{table}
    \centering%
\centering
\begin{tabular}{|r||r|r|r||r|r|r||}
\cline{2-7}
\multicolumn{1}{c||}{} &\multicolumn{3}{c||}{Points on $DH$} & \multicolumn{3}{c||}{$|DH|/r^{2/3}$} \\
\hline
$r$ & Min & Max & Avg. & Min & Max & Avg. \\
\hline
\hline
10 & 15 & 20 & 16 & 3.750 & 5.000 & 4.143 \\
100 & 68 & 85 & 74 & 3.238 & 4.048 & 3.555 \\
1,000 & 333 & 360 & 345 & 3.330 & 3.600 & 3.450 \\
10,000 & 1,568 & 1,636 & 1,603 & 3.379 & 3.526 & 3.456 \\
100,000 & 7,355 & 7,521 & 7,442 & 3.415 & 3.492 & 3.455 \\
1,000,000 & 34,358 & 34,720 & 34,535 & 3.436 & 3.472 & 3.454 \\
10,000,000 & 159,879 & 160,674 & 160,317 & 3.445 & 3.462 & 3.454 \\
\hline
\end{tabular}

    \caption{Experimental results on the size of the discrete-hull of a disk of radius $r$. }
    \tbllab{results_dh}
\end{table}

\begin{table}
    \centering%
\centering
\begin{tabular}{|r||r|r|r||r|r|r||}
\cline{2-7}
\multicolumn{1}{c||}{} &\multicolumn{3}{c||}{Points on $DHS$} & \multicolumn{3}{c||}{$|DHS|/r^{2/3}$} \\
\hline
$r$ & Min & Max & Avg. & Min & Max & Avg. \\
\hline
\hline
10 & 32 & 46 & 40 & 8.000 & 11.500 & 10.123 \\
100 & 216 & 251 & 235 & 10.286 & 11.952 & 11.216 \\
1,000 & 1,075 & 1,278 & 1,217 & 10.750 & 12.780 & 12.173 \\
10,000 & 5,697 & 6,260 & 5,986 & 12.278 & 13.491 & 12.902 \\
100,000 & 28,081 & 29,748 & 28,807 & 13.037 & 13.811 & 13.374 \\
1,000,000 & 133,964 & 139,767 & 137,535 & 13.396 & 13.977 & 13.754 \\
10,000,000 & 637,281 & 657,949 & 647,339 & 13.730 & 14.175 & 13.947 \\
\hline
\end{tabular}

    \caption{Experimental results on the number of lattice points on the boundary of the discrete hull of a disk of radius $r$.}
    \tbllab{results:dhs}
\end{table}

As the results of \tblref{results:dhs} indicate, the ratio between the
average number of lattice points on the boundary of the discrete hull
of a disk of radius $r$ and $r^{2/3}$, appears to be monotone
increasing. On the other hand, by \lemref{dhs:complexity} this ratio
is $O(1)$. Thus, we conjecture that this ratio converges to a
constant, as $r$ tends to infinity.

\subsection*{Acknowledgments}

I wish to thank my thesis advisor, Micha Sharir, for his help in preparing this manuscript. I also wish to thank Pankaj Agarwal, Imre B{\'a}r{\'a}ny, Alon Efrat, and G{\"u}nter Rote for helpful discussions concerning this and related problems.  The author also thank Gill Barequet and the referees for their comments and suggestions.

\BibLatexMode{\printbibliography}

\end{document}